\documentclass[journal,10pt]{IEEEtran}

\usepackage{ifpdf}
\usepackage{graphicx}
\usepackage{cite}
\usepackage{nomencl}
\usepackage{enumitem}
\usepackage{booktabs}
\usepackage{slashbox}

\ifCLASSINFOpdf
\else
\fi
\usepackage{amsmath,amssymb}

\usepackage{algorithmic}
\usepackage{adjustbox}

\usepackage{array}
\usepackage[tablename=Table]{caption}

\usepackage{stfloats}
\usepackage{url}
\usepackage{amsthm}
\usepackage[final]{changes}
\makeatletter
\setdeletedmarkup{\@gobble{#1}}
\makeatother

\newtheorem{theorem}{Theorem}

\newtheorem{corollary}{Corollary}
\newcommand{\RomanNumeralCaps}[1]{\MakeUppercase{\romannumeral #1}}
\usepackage{stackengine,scalerel}

\newcommand\overstar[1]{\ThisStyle{\ensurestackMath{%
			\setbox0=\hbox{$\SavedStyle#1$}%
			\stackengine{0pt}{\copy0}{\kern.2\ht0\smash{\SavedStyle*}}{O}{c}{F}{T}{S}}}}

\begin{document}

\title{Probabilistic Voltage Sensitivity Analysis (PVSA) to Quantify Impact of High PV Penetration on Unbalanced Distribution System}
\author{\IEEEauthorblockN{Sai Munikoti,~\textit{Student Member, IEEE},  Balasubramaniam Natarajan,~\textit{Senior Member, IEEE}, \\ Kumarsinh Jhala,~\textit{Member, IEEE}, Kexing Lai,~\textit{Member, IEEE}
		\thanks{K. Jhala is with the Center for Energy, Environmental, and Economic Systems Analysis in the Energy Systems Division at Argonne National Laboratory (e-mail: kjhala@anl.gov).}
		\thanks{S. Munikoti, K. Lai and B. Natarajan are with Electrical and Computer Engineering, Kansas State University, Manhattan, KS-66506, USA, (e-mail: saimunikoti@ksu.edu, klai@ksu.edu, bala@ksu.edu)}
		\thanks{This work has been submitted to the IEEE for possible publication. Copyright may be transferred without notice, after which this version may no longer be accessible.}}}

\markboth{Preprint submitted to IEEE Transactions on Power Systems}%
{Sai \MakeLowercase{\textit{et al.}}: Bare Demo of IEEEtran.cls for IEEE Journals}

\maketitle

\begin{abstract}
From operational and planning perspective, it is important to quantify the impact of increasing penetration of photovoltaics on the distribution system. Most existing impact assessment studies are scenario based where derived results are scenario specific and not generalizable. Moreover, stochasticity in temporal behavior of spatially distributed PVs requires large number of scenarios that increases with the size of the network and the level of penetration. Therefore, we propose a new computationally efficient analytical framework of voltage sensitivity analysis that allows for stochastic analysis of voltage change due to random changes in PV generation. We first derive an analytical approximation for voltage change at any node of the network due to change in power at other nodes in an unbalanced distribution network. Quality of this approximation is reinforced via bounds on the approximation error. Then, we derive the probability distribution of voltage change at a certain node due to random changes in power injections/consumptions at multiple locations of the network. The accuracy of the proposed PVSA is illustrated using a modified version of IEEE 37 bus test system. The proposed PVSA can serve as a powerful tool for proactive monitoring/control and ease the computational burden associated with perturbation based cybersecurity mechanisms.
\end{abstract}


%
\IEEEpeerreviewmaketitle
\vspace{-0.3cm}
\section{Introduction}

\IEEEPARstart{T}{he} power grid is undergoing significant changes with the integration of renewable energy resources, electric vehicles and active consumers. Massive deployments of rooftop photovoltaic (PV)
generation and demand response programs to incentivize consumers for  peak load shaving are emerging across communities around the world. Despite a variety of benefits, high PV penetration imposes significant challenges on control and operation of distribution systems, including ($1$) voltage stability affected by the increase in underlying uncertainty due to intermittent power characteristics; ($2$) complexity of the system associated with bidirectional power flow, and ($3$) unbalanced characteristics due to variable number and size of PV installations on the three phases \cite{malekpour2015radial}. One approach to improve control and management is to leverage the information aggregated from sensors and devices at the grid edge. This information along with classical load flow algorithms are used to implement various control strategies \cite{aghatehrani2012reactive,weckx2014voltage}. However, such proactive control strategies rely on cumbersome computation of sensitivity matrices, which need to be recomputed whenever the state changes and do not incorporate the spatio-temporal stochasticity of the sources. Additionally, there are multiple recent efforts on enhancing cyber security of the grid edge devices that rely on moving target paradigms. The moving target detection (MTD) approaches involve perturbing the system (e.g. changing $P,Q$ set points of inverter) and observing the response of the system to identify malicious actions \cite{rahman2014moving}, \cite{liu2018hidden},\cite{zhang2019analysis}. These MTD strategies rely on multiple computationally cumbersome load flow simulations to quantify expected behavior. Both proactive control and practical implementation of MTD strategies demand a sensitivity analysis method that is computationally efficient, scalable and can incorporate uncertainties.

\textbf{Related work}: Voltage sensitivity analysis (VSA), which quantifies the voltage variation at a given node due to power changes at other locations of the network, can be used as an effective tool to quantify the impacts of PV variations and MTD related intentional perturbations on the voltage stability across the network. Methodologies for VSA can be broadly divided into two categories, i.e., numerical and analytical. Numerical VSA methods rely on algorithms to give approximate solutions such as Newton-Raphson (NR) load flow method and perturb-and-observe method, which suffer from high computational cost and lack of insights on the system states. Many prior research efforts have examined the performance of numerical sensitivity analysis methods as it relates to regulating voltage in a power system with distributed generators (DGs)\cite{aghatehrani2011reactive, aghatehrani2012reactive, yan2012voltage}, \deleted{\cite{weckx2014voltage}}, \cite{valverde2013model, samadi2014coordinated}, and its drawback in terms of computational efficiency is repeatedly unveiled in these literature. \added{For instance, authors in \cite{aghatehrani2012reactive} present a reactive power control method based on voltage sensitivity analysis for mitigating voltage variations in PV integrated distribution systems. Specifically, a new set point for reactive power is computed with varied active power injection/consumption at other nodes, using Newton-Raphson method for load flow calculation.}\deleted{In \cite{aghatehrani2012reactive,weckx2014voltage}, the Newton-Raphson method is used to control the voltage fluctuation in a PV integrated system.}\added{ \cite{yan2012voltage} proposes a method for analyzing voltage variations due to PV generation fluctuations in unbalanced distribution grids, considering a variety of factors. 
However, its dependency on the inefficient simulation method limits its applications in large scale distribution networks.
In \cite{valverde2013model}, a model predictive control method is proposed to coordinate the active and reactive power of DGs and on-load tap changing transformers set-points for voltage regulation. However, the online update of sensitivity matrix of bus voltages is not realistic using the proposed method due to the high computational burden.} \deleted{In \cite{valverde2013model}, a model predictive controller is used along with the sensitivity matrix to regulate voltages.} 
Authors in \cite{newaz2019coordinated} proposed a centralized coordinated voltage control algorithm for distribution systems with (DGs). Here, Newton-Raphson method is used to examine DG's effect on the voltage stability of a certain node due to reactive power injection at different nodes across the network. \cite{kang2019reactive} proposes a new reactive power management method for minimizing voltage variation \added{in both steady state and transient conditions} due to DER integration. Here, the reactive power of each DER is controlled \added{by exploiting the numerical relationship between variations of voltage and reactive power, based on the traditional VSA method.}
Author in \cite{huang2019day}, \added{develops an optimization model for electric vehicle management based on VSA approaches. Still, the requirements of iterative executions of power flow calculations and optimization models hinder its application in real-world scenarios.
}\deleted{used voltage sensitivity approaches along with power flow to minimize battery degradation in charging/discharging and maximize peak load shifting}\added{Further, an active distribution network management approach is proposed in \cite{ding2016distributed, pezeshki2017probabilistic} for maximization of PV hosting capacity. The approach involves adjusting switching capacitors and voltage regulator taps. In this case, thousands of scenarios are incorporated to address the uncertainties that reveal the huge computational burden of VSA in the presence of renewable energy resources. 
To summarize,} most of these numerical approaches involve computationally expensive load flow algorithms \added{or some kind of trade-off that negatively impacts performance}, thereby limiting their applicability in large scale distribution systems with uncertainties \cite{afolabi2015analysis}.

To overcome the drawbacks of numerical methods, there are some limited analytical approaches for VSA that have been proposed. In \cite{brenna2010voltage}, a new \added{sensitivity matrix is derived analytically, relating voltage magnitude with reactive power change. Then, the sensitivity product is maximized to obtain the optimal generator that has the greatest influence on the voltage of the critical node.} Similarly, in \cite{zad2016centralized}, \added{an algorithm based on the sensitivity analysis
has been designed which optimally manages active and reactive powers of
DGs in order to keep the system voltages inside the limits. Here, instead of repeating load flow calculation to solve the optimization problem, a sensitivity matrix is used to conduct load flow computation in a non-iterative manner, reducing the computational burden significantly. However, the algorithms proposed both in \cite{brenna2010voltage} and \cite{zad2016centralized} are not properly validated with standard test systems.} 
\deleted{the complex power of DGs is controlled to keep the voltage within safe limits, but the proposed VSA method in this work is not validated with any standard test system.}Authors in \cite{klonari2016application}, have taken a probabilistic approach where smart meter data is used along with sensitivity analysis to define boundary values of various operation indices. Here, the real and reactive power consumption of houses are assumed to be independent which is not the case in reality and the proposed approach doesn't account for unbalanced load conditions. In \cite{mugnier2016model}, authors have computed voltage sensitivities by formulating an over-determined system of linear equations constructed solely using measurements of nodal power injections and voltage magnitudes. Similarly, \cite{valverde2018estimation} uses smart meter data with a linear regression model for predicting the voltage change but both \cite{valverde2018estimation, mugnier2016model} rely on the availability of data and monitoring infrastructure. \added{Authors in \cite{weckx2014voltage} obtain load dependent voltage sensitivity factors and develop linearized load flow model based on historical smart meter data comprising of load and voltage profiles, without leveraging any  grid topology information. This work relies heavily on the availability of smart meter data at customer level and data needs to be recollected whenever the network gets reconfigured.} In a nutshell, existing analytical approaches are not generalized enough for analysis of large scale unbalanced distribution systems with stochastic behavior. Therefore, in our prior work \cite{jhala2017probabilistic}, an analytical bound for voltage sensitivity is derived for single phase balanced distribution network. Building off our preliminary work, in this paper, we propose an analytical VSA for a general case of three phase unbalanced distribution system where stochastic power fluctuations can simultaneously occur at multiple nodes of the network. This extension presents many challenges as power change in any one phase impacts the voltage in all the phases.\\

\textbf{Contributions}:
The analytical VSA strategy proposed in this paper not only addresses the computational shortcomings of numerical approaches but systematically incorporates uncertainties. The key contributions of our work include :
\begin{itemize}[leftmargin=*]	
\item An analytical approximation of voltage change due to power change at multiple nodes in an unbalanced distribution network is derived in Section \RomanNumeralCaps{2} (Corollary 1).
\item We derive an upper bound on the approximation error associated with the analytical approximation to further demonstrate its accuracy in Section \RomanNumeralCaps{3} (Corollary 2).
\item To systematically incorporate the stochasticity of power variations, the theoretical probability distribution of voltage change due to random power changes at multiple actor nodes is derived in Section \RomanNumeralCaps{4} (Theorem 2). The resulting PVSA can enable proactive voltage monitoring for identifying voltage violations, due to fluctuations in PV generation \cite{jhala2019data}.  
\item The computational complexity of the proposed method is $O(1)$, i.e., pretty much constant time for execution regardless of network size. Classical NR method has a complexity of $O(n^{3})$, i.e., the execution time scales cubically with the size of the network. 
\end{itemize}
\vspace{-0.10cm}
\section{Analytical approximation of VSA}
This section introduces an analytical approach to VSA for three phase unbalanced power distribution system. Changes in real or reactive power at any phase of a bus results in voltage changes at all phases across all nodes of the distribution system. Nodes where power changes are referred to as actor nodes ($A$), and the nodes where voltage change is monitored are referred to as observation nodes ($O$). This work assumes that the source bus is a slack bus and the \added{load is modeled as constant power load} with star configuration, which serves as an example for illustration. In our preliminary work \cite{8973956}, we derive an analytical approximation for voltage change at an observation node due to the power change at an actor node. The main result is stated in Theorem 1.
\begin{figure}[h!]
	\centering
	\includegraphics[width = 8.9cm,height=4.5cm]{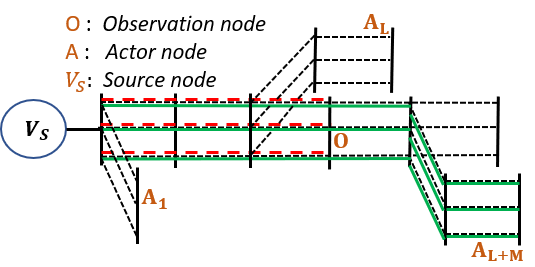}
	\caption{Example network with multiple actor nodes}
	\label{fig:1}
\end{figure}
\begin{theorem}
	For an unbalanced power distribution system, change in complex voltage ($\Delta V_{OA}$) at an observation node ($O$) due to change in complex power of an actor node ($A$) can be approximated by 
	\begin{equation}	
	\begin{bmatrix}
	\Delta V_{O}^{a} \\[13pt]
	\Delta V_{O}^{b} \\[13pt]
	\Delta V_{O}^{c} 
	\end{bmatrix} \approx -	
	\begin{bmatrix}
\frac{\Delta S_{A}^{a\star}Z_{OA}^{aa}}{ V_{A}^{a\star}} + \frac{\Delta S_{A}^{ b\star}Z_{OA}^{ab}}{ V_{A}^{b\star}}+ \frac{\Delta S_{A}^{ c\star}Z_{OA}^{ac}}{V_{A}^{c\star}} \\[7pt]
\frac{\Delta S_{A}^{a\star}Z_{OA}^{ba}}{ V_{A}^{a\star}} + \frac{\Delta S_{A}^{b\star}Z_{OA}^{bb}}{ V_{A}^{b\star}}+ \frac{\Delta S_{A }^{c\star}Z_{OA}^{bc}}{ V_{A}^{c\star}}  \\[7pt]
\frac{\Delta S_{A}^{ a\star}Z_{OA}^{ca}}{ V_{A}^{a\star}} + \frac{\Delta S_{A}^{b\star}Z_{OA}^{cb}}{ V_{A}^{b\star}}+ \frac{\Delta S_{A}^{ c\star}Z_{OA}^{cc}}{ V_{A}^{c\star}}
\end{bmatrix}
	\label{eq:1a} 
	\end{equation}
	where $a,b$ and $c$ represent the three phases, which are used throughout the paper. $V_{A}^{a\star}$ and $\Delta S_{A}^{a}$ represent complex conjugate of voltage at phase $a$ and complex power change at actor node $A$, respectively; $Z$ denotes the impedance matrix including self and mutual line impedance of the shared path between observation node and actor node from the source node. Fig.\ref{fig:1} depicts an unbalanced three phase distribution system which will be used for illustrations of Theorems. The red lines in Fig. \ref{fig:1} represent the shared paths between actor node $A_{L+M}$ and observation node $O$, from the source node.
\end{theorem}
\begin{proof}
	For brevity, the proof is omitted here and the readers are directed to \cite{8973956} for details.
\end{proof}
\subsection{Multiple actor nodes}
With increasing penetration of DERs at different locations across the grid, it is important to extend the VSA to include multiple actor nodes, resulting in Corollary 1:
\begin{corollary}
	For an unbalanced power distribution system, change in complex voltage $\Delta V_{O}$ at an observation node ($O$) due to change in complex power at multiple actor nodes can be approximated by 
	\begin{equation}	
	\begin{bmatrix}
	\Delta V_{O}^{a} \\[15pt]
	\Delta V_{O}^{b} \\[15pt]
	\Delta V_{O}^{c} 
	\end{bmatrix} \approx- \sum_{A \epsilon \tilde{A}}  \left(	
	\begin{bmatrix}
\frac{\Delta S_{A}^{a\star}Z_{OA}^{aa}}{ V_{A}^{a\star}} + \frac{\Delta S_{A}^{ b\star}Z_{OA}^{ab}}{ V_{A}^{b\star}}+ \frac{\Delta S_{A}^{ c \star}Z_{OA}^{ac}}{ V_{A}^{c\star}} \\[8pt]
\frac{\Delta S_{A}^{ a\star}Z_{OA}^{ba}}{ V_{A}^{a\star}} + \frac{\Delta S_{A}^{ b\star}Z_{OA}^{bb}}{ V_{A}^{b\star}} + \frac{\Delta S_{A }^{c\star}Z_{OA}^{bc}}{ V_{A}^{c\star}}  \\[8pt]
\frac{\Delta S_{A}^{ a\star}Z_{OA}^{ca}}{ V_{A}^{a\star}} + \frac{\Delta S_{A}^{ b\star}Z_{OA}^{cb}}{ V_{A}^{b\star}}+ \frac{\Delta S_{A}^{ c\star}Z_{OA}^{cc}}{ V_{A}^{c\star}}
\end{bmatrix}\right)
\label{eq:2} 
\end{equation}
	where $\tilde{A}$ is the set of all actor nodes.
\end{corollary}
\begin{proof}
Voltage at an observation node can be written in terms of source voltage and voltage drop across the lines (edges) between source node and observation node as
\begin{equation}
\pmb{V_{O}} = \pmb{V_{S}} - \sum_{e\epsilon E_{o}} \pmb{Z_{e}I_{e}}.
\label{eq:3}
\end{equation}
\hspace{1cm}
$\pmb{V_{O}}=\begin{bmatrix}
V_{O}^{a} \\[3pt]
V_{O}^{b} \\[3pt]
V_{O}^{c} 
\end{bmatrix}$ and  $\pmb{Z_{e}}=\begin{bmatrix}
Z_{e}^{aa} &Z_{e}^{ab} & Z_{e}^{ac} \\[3pt]
Z_{e}^{ba} &Z_{e}^{bb} & Z_{e}^{bc} \\[3pt]
Z_{e}^{ca} &Z_{e}^{cb} & Z_{e}^{cc} 
\end{bmatrix}$ \\ \\ 
where $\pmb{I_{e}}$ and $\pmb{Z_{e}}$ are the current vector and line impedance matrix for edge $e$, respectively. $E_{o}$ is set of all edges between the source node and observation node $O$. \added{It is to be noted that for three phase four wire distribution system, the line impedance matrix is a $4\times 4$ matrix, which accounts for a neutral conductor along with the conductors of three phases. Therefore, to incorporate this system in our framework, an equivalent $3\times3$ impedance matrix needs to be computed by using the Carson's method followed by Kron's reduction. The Kron's formula for each element $Z_{ij}$ of the impedance matrix $Z_{e}$ is given by \cite{kersting1994distribution},
\begin{equation}
Z_{ij}^{'}= Z_{ij}-\frac{Z_{in}Z_{nj}}{Z_{nn}}
\label{eq:3a}
\end{equation}
where, $Z_{in}$ and $Z_{nj}$ are the mutual impedance of the conductors at phase $i$ and $j$ (with respect to the neutral conductor), respectively. $Z_{nn}$ is the self impedance of the neutral conductor, and $Z_{ij}$ is the mutual impedance between phase $i$ and phase $j$.
Thus, (\ref{eq:3a}) can generate the equivalent line impedance matrix $Z_{e}$ for a four wire system, which can be plugged in (\ref{eq:3}) to compute the voltage at the observation node in a four-wire distribution network. Now, }let $S_{n}$ be the complex power drawn or injected at node $n$ and $V_{n}^{\star}$ be the complex conjugate of voltage at node $n$. The current flowing through a particular edge $e$ of $E_{o}$ can be written as
\begin{equation}
\pmb{I_{e}}=\begin{bmatrix}
I_{e}^{a} \hspace{0.1cm} I_{e}^{b} \hspace{0.1cm} I_{e}^{c}  
\end{bmatrix}^{T}= \sum_{n\epsilon N_{e}}\begin{bmatrix}
\frac{S_{n}^{a\star}}{V_{n}^{a\star} } \hspace{0.1cm} \frac{S_{n}^{b\star}}{V_{n}^{b\star} } \hspace{0.1cm} \frac{S_{n}^{c\star}}{ V_{n}^{c\star} }
\end{bmatrix}^{T},
\label{eq:3b}
\end{equation}
where $N_{e}$ is the set of all nodes $n$ for which edge $e$ is between node $n$ and the source node. In other words, power from the source node to all the nodes in the set $N_{e}$ flows through edge $e$. Therefore, current in edge $e$ will be affected by the power change at nodes $ n \epsilon N_{e}$. \deleted{It is important to note that the current is expressed in equation (\ref{eq:3b}) by employing $S=V^{*}I$ form, where current is assumed to be reference and voltage is measured with respect to it. More importantly, we are interested in the magnitude of voltage change which is invariant to both the representations of complex power $S$, i.e., $S=V^{*}I$ and $S=VI^{*}$. Now,} 
\deleted{Thus}Then, voltage at the observation node can be expressed as,
\begin{equation}
\pmb{ V_{O}} =\pmb{V_{S}} -\sum_{e\epsilon E_{o}} \pmb{Z_{e}} \sum_{n\epsilon N_{e}} 
\begin{bmatrix}
\frac{S_{n}^{a\star}}{V_{n}^{a\star} } \hspace{0.1cm} \frac{S_{n}^{b\star}}{V_{n}^{b\star} } \hspace{0.1cm} \frac{S_{n}^{c\star}}{ V_{n}^{c\star} }
\end{bmatrix}^{T}.
\label{eq:4} 
\end{equation}
When power consumption of node $n$ changes from $S_{n}$ to $S_{n}+ \Delta S_{n} $, the voltage of node $n$ will change from $V_{n}$ to $V_{n} + \Delta V_{n}$ and consequently voltage at observation node will change to $\pmb{V_{O}^{'}}$. The effective voltage change at observation node $\pmb{\Delta V_{O}}$ $(i.e.,\pmb{V_{O}-V_{O}^{'})}$ can then be written as: 

\begin{equation}
\begin{split}
\pmb{\Delta V_{O}} 
& =\sum_{e\epsilon E_{o}} \pmb{Z_{e}} \left( \sum_{n\epsilon N_{e}} \begin{bmatrix}
\frac{S_{n}^{a\star}\Delta V_{n}^{a\star}-\Delta S_{n}^{a\star} V_{n}^{a\star}}{ V_{n}^{a\star} ( V_{n}^{a\star} + \Delta V_{n}^{a\star} )}  \\[11pt]
\frac{S_{n}^{b\star}\Delta V_{n}^{b\star}-\Delta S_{n}^{b\star} V_{n}^{b\star}}{V_{n}^{b\star} ( V_{n}^{b\star} + \Delta V_{n}^{b\star} )}  \\[11pt]
\frac{S_{n}^{c\star}\Delta V_{n}^{c\star}-\Delta S_{n}^{c\star} V_{n}^{c\star}}{V_{n}^{c\star} ( V_{n}^{c\star} + \Delta V_{n}^{c\star})} 
\end{bmatrix} \right) 
\label{eq:5} 
\end{split}.
\end{equation}
In practice, voltage changes are typically small compared to actual node voltage. Hence, it is reasonable to assume that $\Delta V_{n}^{a\star}/(V_{n}^{a\star}+\Delta V_{n}^{a\star}) \to 0 $. Hence, (\ref{eq:5}) is approximated as,
\begin{equation}
\begin{split}
\pmb{\Delta V_{O}}
&=\sum_{e\epsilon E_{o}} \pmb{Z_{e}} \left( \sum_{n\epsilon N_{e}} \begin{bmatrix}
\frac{-\Delta S_{n}^{a\star} }{V_{n}^{a\star}  +  \Delta V_{n}^{a\star}}  \\[8pt]
\frac{-\Delta S_{n}^{b\star} }{V_{n}^{b\star}  + \Delta V_{n}^{b\star}}  \\[8pt]
\frac{-\Delta S_{n}^{c\star} }{V_{n}^{c\star}  + \Delta V_{n}^{c\star}} 
\end{bmatrix} \right) \\
& = \sum_{e\epsilon E_{o}} \pmb{Z_{e}} \left(\sum_{n\epsilon N_{e}} \pmb{I_{n}} \right), 
\end{split}
\label{eq:6} 
\end{equation}
where, $\pmb{I_{n}} = \begin{bmatrix}
\frac{-\Delta S_{n}^{a\star} }{V_{n}^{a} + \Delta V_{n}^{a\star} } \hspace{0.1cm} \frac{-\Delta S_{n}^{b\star}}{V_{n}^{b\star}  + \Delta V_{n}^{b\star} } \hspace{0.1cm} \frac{-\Delta S_{n}^{c\star}}{V_{n}^{c\star} + \Delta V_{n}^{c\star}}
\end{bmatrix}^{T} $. \\ \\
Let us assume there are $L+M$ actor nodes such that there are $L$ nodes between the source node and observation node $O$ and $M$ nodes between observation node and last actor node of the network as shown in the Fig. 1. The nodes are arranged in such a way that the set $E_{O}\cap E_{A_{1}} $ has minimum elements (edges) and the sets  $E_{O}\cap E_{A_{L+1}} $ to $E_{O}\cap E_{A_{L+M}} $ have same and maximum number of edges. This is mathematically represented as,
\begin{equation} 
\begin{split}
& |E_{O}\cap E_{A_{1}}|\leq |E_{O}\cap E_{A_{2}}| \hdots \leq |E_{O}\cap E_{A_{L}}|\\
& \leq |E_{O}\cap E_{A_{L+1}}|=|E_{O}\cap E_{A_{L+2}}| \hdots =|E_{O}\cap E_{A_{L+M}}|
\end{split} 
\label{eq:7}
\end{equation} 
where $|E_{O}\cap E_{A_{1}}|$ denotes the cardinality of set $E_{O}\cap E_{A_{1}}$. On dividing set $E_{O}$ into $L+1$ subsets as,
\begin{equation} 
\begin{split} E_{O}=
& |E_{O}\cap E_{A_{1}}|\cup|E_{O}\cap (E_{A_{2}}- E_{A_{1}})|\cup \hdots \\
& |E_{O}\cap (E_{A_{L+1}}- E_{A_{L}})| \\
& = \bigcup\limits_{l=1}^{A_{L+1}}E_{O} \cap (E_{A_{l}}- E_{A_{l-1}}) 
\end{split} 
\label{eq:8}
\end{equation}
since $E_{O} \cap (E_{A_{L}}- E_{A_{L-1}})=\phi$ for $A_{L}= A_{L+2}$ or greater. Using this, (\ref{eq:6}) can be be expressed as,
\begin{equation}
\begin{split}
\pmb{\Delta V_{O}}
& =\sum_{l=1}^{L+1} \sum_{e\epsilon E_{O}\cap E_{A_{l}}-E_{O}\cap E_{A_{l-1}}} 
\left(\pmb{Z_{e}} \pmb{I_{n}} \right)\\
& =\sum_{n=A_{1}}^{A_{L}} \left( \sum_{e\epsilon E_{o}\cap E_{n}} 
 \pmb{Z_{e}}\right) \pmb{I_{n}}. 
\end{split}
\label{eq:9} 
\end{equation} 
When power injection/consumption changes at the actor node $n$, current flowing through the edges changes for all edges of the set $E_{n}$. However, voltage drop across the edges between source node and observation node, changes only for edges that belongs to subset $E_{n} \cap E_{o}$. Taking the sum of the impedance across all such edges, reduces (\ref{eq:9}) to the following form:
\begin{equation} 
\begin{split}
\pmb{\Delta V_{O}}=\sum_{n=A_{1}}^{A_{L}}
\pmb{Z_{on}} \pmb{I_{n}} .
\end{split} 
\label{eq:10}
\end{equation}
where $\pmb{Z_{on}}=\sum_{e\epsilon E_{o}\cap E_{n}} \pmb{Z_{e}}$ is the impedance matrix and its elements are computed by the summation of the impedances of shared paths between all actor nodes and observation node from source node. By expanding $\pmb{I_{n}}$ and $\pmb{Z_{on}}$, (\ref{eq:10}) can be decomposed into real and imaginary parts. The real part $\pmb{\Delta V_{OA}^{a,r}}$ for a particular phase (suppose $a$) can be written as:
\normalsize
\begin{equation}
\begin{split}
\pmb{\Delta V_{OA}^{a,r}} = -
\frac{(\Delta P_{A}^{a}R_{OA}^{aa} + \Delta Q_{A }^{a}X_{OA}^{aa})(V_{A}^{a,r}+\Delta V_{A }^{a,r})}{(V_{A}^{a,r}+ \Delta V_{A }^{a,r})^{2}+(V_{A}^{a,i}+ \Delta V_{A }^{a,i})^{2}} + \\ 
\frac{(\Delta P_{A }^{a}X_{OA}^{aa}-\Delta Q_{A}^{a}R_{OA}^{aa})(V_{A }^{a,i}+\Delta V_{A}^{a,i})}{(V_{A}^{a,r}+ \Delta V_{A }^{a,r})^{2}+(V_{A}^{a,i}+ \Delta V_{A }^{a,i})^{2}}- \dots \\[6pt]
\end{split}
\label{eq:10a}
\end{equation}
In a distribution network, the magnitude of voltage change is usually very small, which can be used to approximate the real part of voltage change as:
\begin{equation}
\begin{split}
\pmb{\Delta V_{OA}^{a,r}} \approx -
\frac{(\Delta P_{A}^{a}R_{OA}^{aa}+\Delta Q_{A}^{a}X_{OA}^{aa})(V_{A }^{a,r})}{(V_{A}^{a,r})^{2}+(V_{A}^{a,i})^{2}} + \\ 
\frac{(\Delta P_{A}^{a}X_{OA}^{aa}-\Delta Q_{A}^{a}R_{OA}^{aa})(V_{A}^{a,i} )}{(V_{A}^{a,r})^{2}+(V_{A}^{a,i})^{2}}- \dots \\[6pt]
\end{split}
\label{eq:10b}
\end{equation}
In a similar way, the imaginary part of voltage change can also be approximated. The aggregation of $\Delta P_{A},\Delta Q_{A}$ as $\Delta S_{A}$ and $R_{A}, X_{A}$ as $Z$  gives 
\begin{equation}
\Delta V_{O}^{a} \approx -\sum_{A=A_{1}}^{A_{L}} \left(\frac{\Delta S_{A}^{a\star}Z^{aa}}{V_{A}^{a}} + \frac{\Delta S_{A}^{b\star}Z^{ab}}{V_{A}^{b}}+ \frac{\Delta S_{A}^{c\star}Z^{ac}}{V_{A}^{c}}\right)
\label{eq:11}
\end{equation}
Repeating the same procedure for all the three phases, yields the   voltage change approximation as stated in Corollary 1.
\end{proof}
\subsection{Validation of VSA for multiple actor nodes}
This section verifies the derived analytical approximation of VSA using a modified IEEE 37-node test system. The test system is shown in Fig. \ref{fig:2}, and it is used for the validation of all theoretical approximations proposed in this work. This test network is selected due to its highly unbalanced load and has been used by various researchers in the past for validation \cite{khushalani2007development}. The nominal voltage of the test system is $4.8$ kV with bus $1$ as source. Classical NR method is used as a baseline method for validating our proposed methods.
\added{Along with the IEEE 37-node network, we also employ a larger IEEE 123-node test network for evaluating the proposed method as shown in Fig \ref{fig:2b}. This network is particularly selected due to its highly unbalanced characteristics, consisting of both single and three phase loads. The rated voltage of the test system is $4.16$ kV.}

The accuracy of the VSA approximation for multiple actor node case, is \added{first} evaluated \added{in the 37-node network} by simulating \deleted{the} \added{a} scenario assuming $22,17,14,8,7$ as actor nodes. The power changes at these actor nodes will occur simultaneously at different phases, which is about $50\%$ of their rated load as tabulated in Table \ref{tab:1}. 
\begin{table}
	\centering
	\caption{Power change across different actor nodes}
	\begin{tabular}{ |c|c|c|c|} 
		\hline
		Actor nodes & Phase & Rated power (kVA) & New power (kVA) \\ 
		\hline
		22 & c & 42 + j21 & 63 + j21 \\ 
		\hline
		17 & b & 42 + j21 & 63 + j31 \\ 
		\hline
		14 & c & 84 + j42 & 126 + j21 \\ 
		\hline
		8 & a & 42 + j21 & 63 + j31 \\ 
		\hline
		7 & c & 84 + j42 & 126 + j21 \\ 
		\hline
	\end{tabular}
	\label{tab:1}
\end{table}	
Fig. \ref{fig:3} shows the voltage change at various observation nodes. It can be observed that the errors between the analytical approximation and simulated voltage change are negligible as the maximum deviation is in the range of $10^{-4}$ pu. \added{The absolute value of error average over all the observation nodes is 0.000196 pu, which is very low.} As expected, the magnitude of voltage change increases as the observation node moves away from the source node. This is due to the increase in the length of the shared path between observation node and the actor nodes from the source node. Furthermore, the voltage change remains constant for the range of observation nodes from $25$ to $37$. This is due to the constant length of the shared paths between the actor nodes and these observation nodes, which can be observed from Fig. \ref{fig:2}. \added{To summarize, the accuracy of voltage estimation can be ensured regardless of the relative distance between actor and observation nodes.} An additional inference is that the effect of power variation in the voltage change of phase $a$ seems to be less compared to phase $c$. This is because the majority of selected actor nodes belong to phase $c$ as mentioned in Table \ref{tab:1}. To further quantify the quality of the result in Corollary 1, the error associated with the approximation is analyzed in the next section.  

\begin{figure}[h!]
\centering
	\includegraphics[width = 8cm,height=3.0cm]{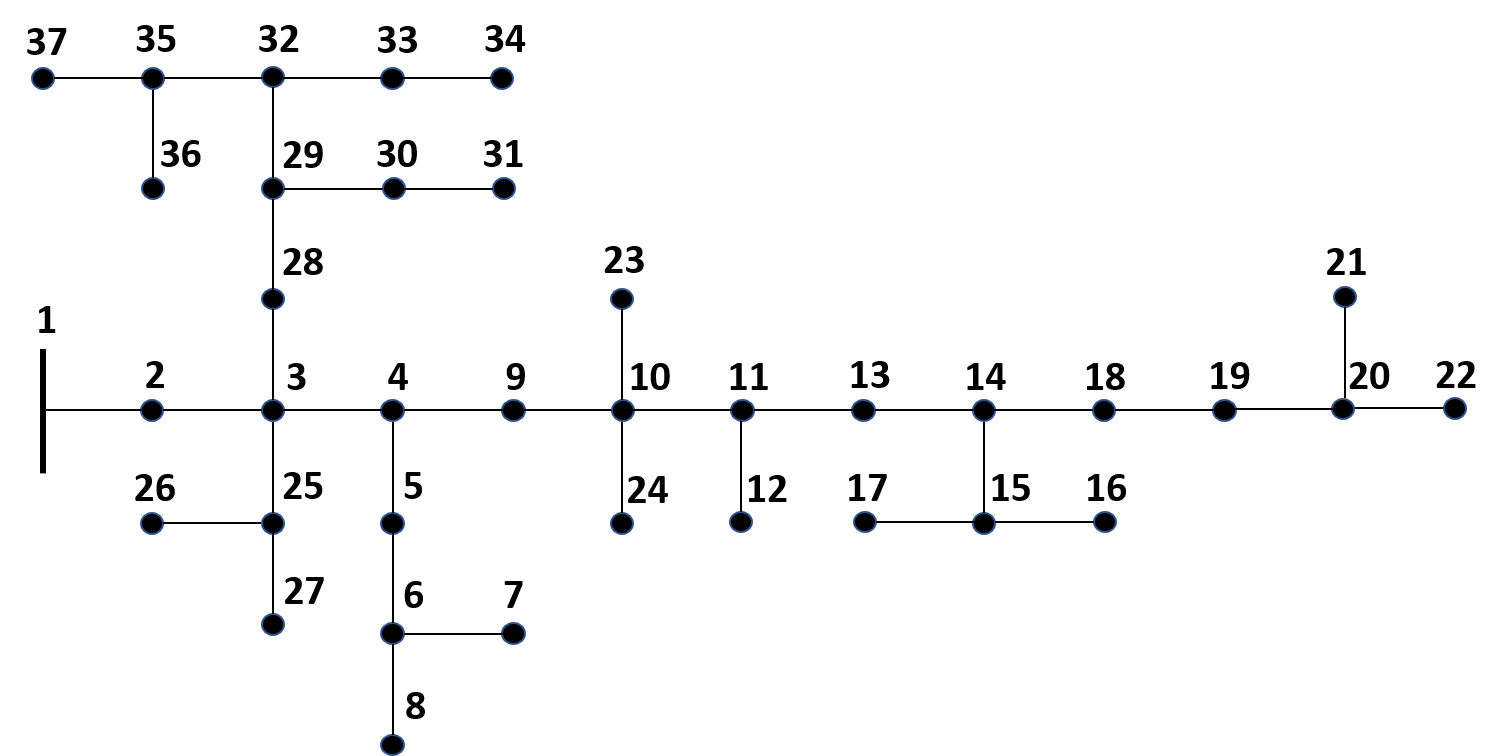}
	\caption{Modified IEEE 37 node network}
	\label{fig:2}
\end{figure}
\begin{figure}[h!]
\centering
	\includegraphics[width = 8.8cm,height=5.5cm]{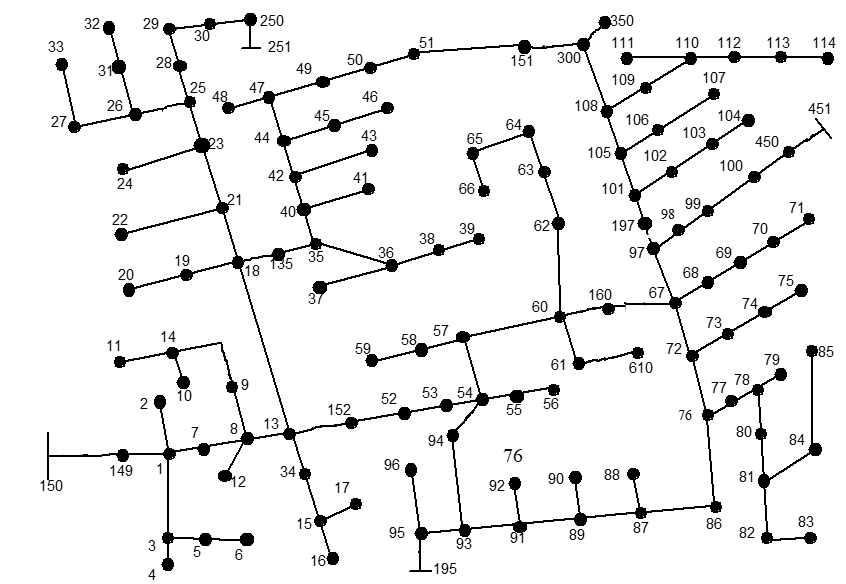}
	\caption{\added{Modified IEEE 123 node network}}
	\label{fig:2b}
\end{figure}
\begin{figure}[h!]
\centering
	\includegraphics[width = 9.0cm,height=5.0cm]{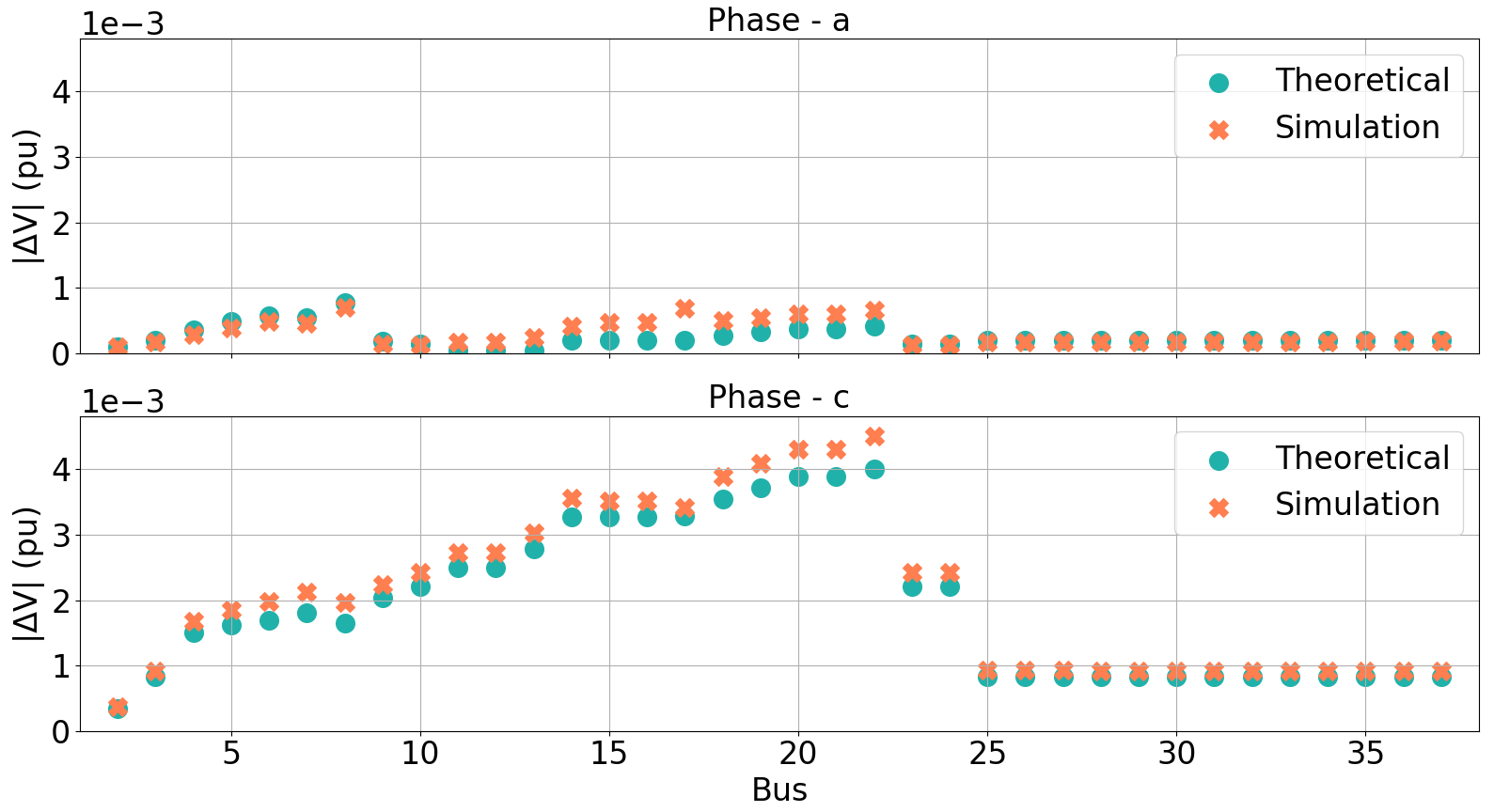}
	\caption{Voltage change on all nodes \added{of 37 node network }due to multiple actor nodes}
	\label{fig:3}
\end{figure}

\begin{figure}[h!]
\centering
	\includegraphics[trim = 0 0 0 0, clip, width = 9.0cm, height=3.0cm]{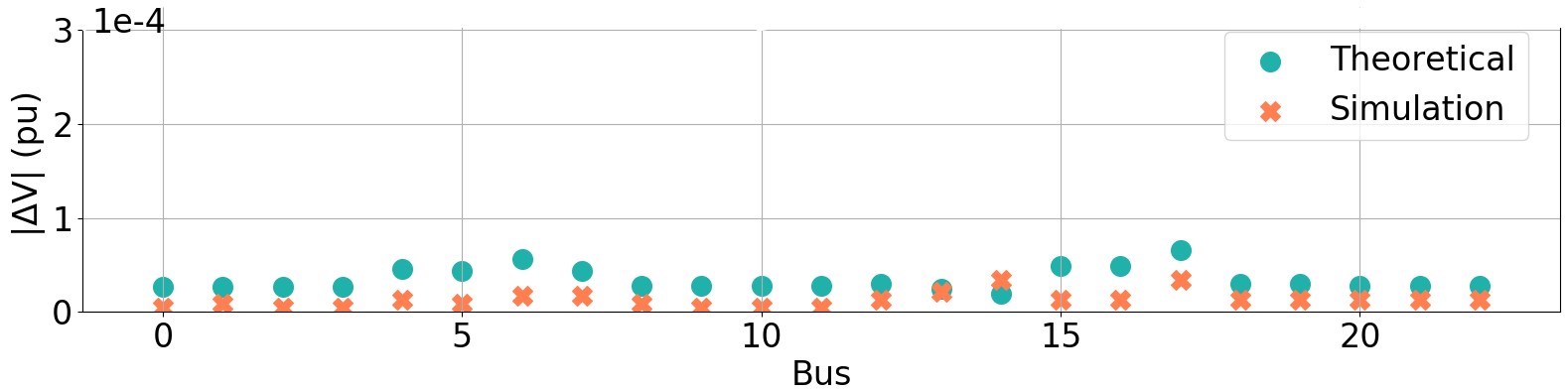}
	\caption{\added{Voltage change at phase a of the selected nodes in 123 node test network due to multiple actor nodes}}
	\label{fig:3b}
\end{figure}

\added{Similarly, the VSA approximation is tested in the 123-node network with $7$ actor nodes, i.e., nodes 7,11,19,28,35,42,68. Like the tests conducted for the IEEE 37-node system, the power changes at these actor nodes occur simultaneously with magnitude equal to 50 \% of their rated load. Fig. \ref{fig:3b} shows the voltage change at various observation nodes. Accurate voltage estimation using the proposed analytical formulation can be observed, as the error is contained within $10^{-4}$ pu.}
\section{Upper bound on approximation error}
As shown in Section \RomanNumeralCaps{2} B, the proposed analytical method approximates the true voltage change for a large range of power variation with very small error magnitude. To further substanticiate the quality of this approximation, Corollary $2$ provides an upper bound for the error. 
\begin{corollary}
	 For an unbalanced power distribution system, the errors in the real ($\Delta V_{e}^{r}$) and imaginary part ($\Delta V_{e}^{i}$) of the voltage change approximation are upper bounded by:
	\begin{equation}
	\begin{split}
	\Delta V_{e}^{r} \leq \sum_{u\epsilon \tilde{U}} \left(\frac{k_{1}^{u}/(1+c_{1}^{u})}{V_{A}^{a,r} } +  \frac{k_{2}^{u}/(1+c_{2}^{u})}{V_{A}^{a,i}}\right) \\
	\Delta V_{e}^{i} \leq  \sum_{u\epsilon \tilde{U}} \left(\frac{k_{2}^{u}/(1+c_{1}^{u})}{V_{A}^{a,r} } +  \frac{k_{1}^{u}/(1+c_{2}^{u})}{V_{A}^{a,i}}\right),	
	\end{split}
	\label{eq:13}
	\end{equation}	
where $k_{1},k_{2},c_{1}, c_{2}$ are parameters dependent on the power change and impedance of the corresponding phases. The set $ \tilde{U}$ contains the self and cross phase terms of the phase where error in voltage change is computed (e.g. it is $aa,ab,ac$ for phase $a$). The voltage change at any phase consist of three components from the three phases. The value of these parameters for phase $a$ are: $k_{1}=\Delta P_{A}^{a} R_{OA}^{aa}-\Delta Q_{A}^{a} X_{OA}^{aa}$, $ k_{2} = \Delta P_{A}^{a}R_{OA}^{aa}-\Delta Q_{A}^{a}X_{OA}^{aa}$, $c_{1}=(V_{A}^{a,i}/V_{A}^{a,r})^2$, $c_{2}=c_{1}^{-1}$. 
\end{corollary} 
\normalsize
\begin{proof}
Firstly, let us recall that the analytical approximation derived in Corollary 1 is based on a legitimate assumption that the value of voltage change ($\Delta V_{A}$) can be ignored compared to the rated voltage ($V_{A}$). In other words, the terms containing $\Delta V_{A}$ in equation (\ref{eq:5}) are removed, which leads to the approximation computed in (\ref{eq:6}). Despite the accurate approximation performed by Corollary 1, as demonstrated in Fig. 3, this simplification incurs the inevitable error. In Corollary 2, we will prove that the incurred error is upper bounded, which ensures the stability of the approximation method proposed in Corollary 1. The approximation error in voltage change at any phase consists of three components, corresponding to three phases. The error $\Delta V_{e}^{aa,r}$ in real part of phase $a$ due to phase $a$ component is the difference of actual (equation (\ref{eq:10a})) and approximation (equ. (\ref{eq:10b})) expressed as,
\begin{equation} 
\footnotesize
\begin{split}
\Delta V_{e}^{aa,r}= \frac{(k_{1})(V_{A}^{a,r}+\Delta V_{A}^{a,r})}{(V_{A}^{a,r}+ \Delta V_{A}^{a,r})^{2}+(V_{A}^{a,i}+ \Delta V_{A a}^{a,i})^{2}} -
\frac{(k_{1})(V_{A}^{a,r})}{(V_{A}^{a,r})^{2}+(V_{A}^{a,i})^{2}} \\
+\frac{(k_{2})(V_{A}^{a,i}+\Delta V_{A}^{a,i})}{(V_{A}^{a,r}+ \Delta V_{A}^{a,r})^{2}+(V_{A}^{a,i}+ \Delta V_{A}^{a,i})^{2}} -
\frac{(k_{2})(V_{A}^{a,i})}{(V_{A}^{a,r})^{2}+(V_{A}^{a,i})^{2}}.
\end{split} 
\label{eq:14}
\vspace{-1.2cm}
\end{equation}
Similar components from phase $b$ and $c$ exist, which together with (\ref{eq:14}) contribute to the error in phase $a$. (\ref{eq:14}) can be further simplified as,
\normalsize
\begin{equation}
\begin{split}
\Delta V_{e}^{aa,r} & =
\frac{(k_{1})(\tau^{r})(V_{A}^{a,r})}{(\tau^{r})^{2}(V_{A}^{a,r})^{2}+(\tau^{i})^{2}(V_{A}^{a,i})^{2}} -
\frac{(k_{1})(V_{A}^{a,r})}{(V_{A}^{a,r})^{2}+(V_{A}^{a,i})^{2}} \\
& +\frac{(k_{2})(\tau^{i})(V_{Aa}^{a,i})}{(\tau^{r})^{2}(V_{A}^{a, r})^{2}+(\tau^{i})^{2}(V_{A}^{a,i})^{2}}-\frac{(k_{2})(V_{A}^{a,i})}{(V_{A}^{a, r})^{2}+(V_{A}^{a,i})^{2}} \\
& = \Delta V_{e1}^{aa,r} + \Delta V_{e2}^{aa,r}
\end{split}
\label{eq:15}
\end{equation}
\normalsize
where $\tau^{r}= 1+ \epsilon^{r}, \tau^{i}= 1+ \epsilon^{i}, \epsilon^{r}=(\Delta V_{A}^{a,r}/{V_{A}^{a,r}}), \epsilon^{i}=(\Delta V_{A}^{a,i}/{V_{A}^{a,i}}) $.
Here, (\ref{eq:15}) consist of two similar error components $\Delta V_{e1}^{aa}$ and $\Delta V_{e2}^{aa}$, which is evaluated separately as:
\begin{equation}
\begin{split}
\Delta V_{e1}^{aa,r}=\frac{k_{1}/(1+c_1)}{ V_{A}^{a,r}} \left(\frac{1}{\tau^{r}}-1\right)
\end{split}
\label{eq:16}
\end{equation}
where $c_{1} = (V_{A}^{a,i}/V_{A}^{a,r})^2.$
As the ratio of change in voltage and rated voltage, i.e., $\epsilon^{r}$ and $\epsilon^{i}$ are typically very small, we can argue the following inequality: 
\begin{equation}
\begin{split}
\epsilon^{r} \leq 1 - \epsilon^{r} \implies
\frac{\epsilon^{r}}{1- \epsilon^{r}} \leq 1 \implies
\frac{1}{\tau^{r}}-1 \leq 1 
\end{split}
\label{eq:17}
\end{equation}
Then, using (\ref{eq:17}), equation (\ref{eq:16}) can be bounded as, 
\begin{equation}
\begin{split}
\Delta V_{e1}^{aa,r}  = \frac{k_{1}/(1+c_{1})}{V_{A}^{a,r}} \left(\frac{1}{\tau^{r}}-1\right) \leq \frac{k_{1}/(1+c_{1})}{V_{A}^{a,r}}
\end{split}
\label{eq:18}
\end{equation}
Similarly, with the same arguments, the upper bound can be derived for second part of (\ref{eq:15}) as, 
\begin{equation}
\begin{split}
\Delta V_{e2}^{aa,r}  = \frac{k_{2}/(1+c_{2})}{V_{A}^{a,i}} \left(\frac{1}{\tau^{i}}-1\right) \leq \frac{k_{2}/(1+c_{2})}{V_{A}^{a,i}}
\end{split}
\label{eq:19}
\end{equation}
Equations (\ref{eq:18}) and (\ref{eq:19}) are combined to arrive at the upper bound on the specific component of the voltage change, contributed from phase $a$.
\begin{equation}
\Delta V_{e}^{aa,r} \leq\frac{k_{1}/(1+c_{1})}{V_{A}^{a,r}} +  \frac{k_{2}/(1+c_{2})}{V_{A}^{a,i}}
\label{eq:20}
\end{equation}
The bound on the other parts of voltage change, which are contributed from phase $b$ and $c$, i.e., $\Delta V_{e}^{ab,r}$ and $\Delta V_{e}^{ac,r}$, is similar in form to (\ref{eq:20}) except for the constants ($k_{1},k_{2},c_{1},c_{2}$) which are dependent on the power and shared path impedance of the corresponding phases. Then, the bound from all the phase terms are added to give the aggregate upper bound on the real part of voltage change in phase $a$, as stated in Corollary 2. The same procedure can be applied to derive the bounds for imaginary part of voltage change in phase $a$. Finally, the bound on the error magnitude can be computed from the bound on real and imaginary part of the error.
\end{proof}
\subsection{Validation of error upper bound}
To evaluate the tightness of the error bound proposed in Corollary 2, the same IEEE 37-node network used in Section \RomanNumeralCaps{2} is adopted. A scenario is simulated where the power drawn by phase $c$ of node $22$ is increased by $21$ kW and voltage change is observed across all nodes. The actual error is computed by taking the difference between the voltage change calculated from numerical load flow method and the analytical approximation. Then, the theoretical error bound is computed using the results of Corollary $2$. Fig. \ref{fig:0b} compares the actual error and the error bound. It can be observed that the theoretical bound is always above the actual error, with sufficient tightness. The bound is relatively tight in cross phases, i.e., phase $a$ and $b$, compared to phase $c$ as power varies in phase $c$ of the network. The error magnitude is relatively large when observation nodes are in the neighborhood of actor node. This is because the shared path impedance between the observation node and actor node from the source node is larger for the neighboring nodes, compared to other nodes in the network. 

\begin{figure}[h!]
	\includegraphics[width = 8.5cm,height=5.5cm]{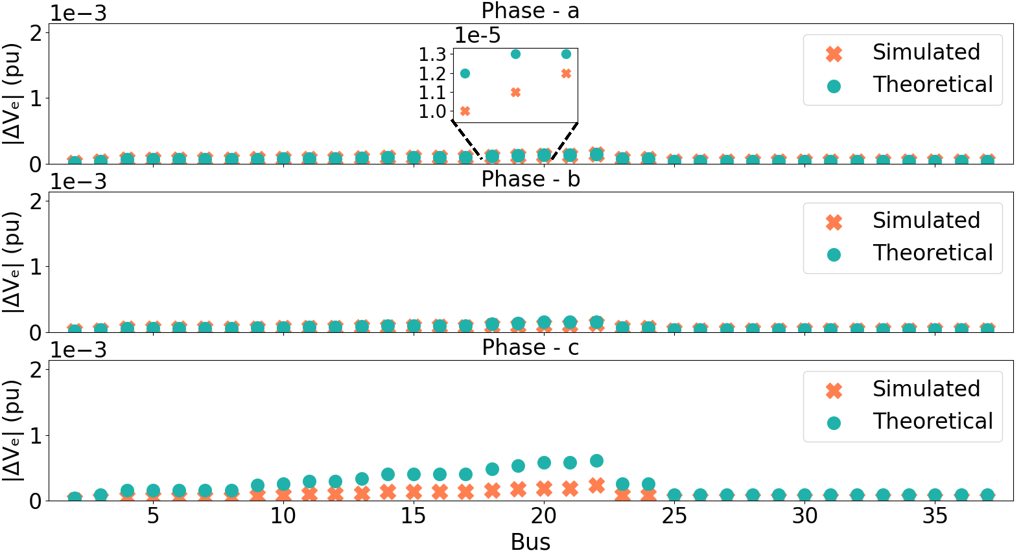}
	\caption{\added{Error bound on the voltage change of all nodes in 37-node network}}
	\label{fig:0b}
\end{figure}
\added{Further, we also check the scalability of the error bound by testing corollary 2 in the larger IEEE 123-node network. Power is increased by 10 kW at phase $a$ of the node 20 and voltage change is monitored for the selected observation nodes. The actual error and theoretical error bound is computed in a similar way as described for the 37-node network, and are shown in Fig. \ref{fig:4b}. Once again the results demonstrate the tightness of error bound. Now, given the error is upper bounded, it}\deleted{The error bound} further increases the credibility of the analytical approximation and allow us to extend the analysis for \added{a} stochastic framework, where the power changes are uncertain. This extension to the probabilistic case is discussed in \added{the} next section.
\begin{figure}[h!]
\centering
	\includegraphics[trim = 0 0 0 0, clip, width = 8.5cm,height=2.5cm]{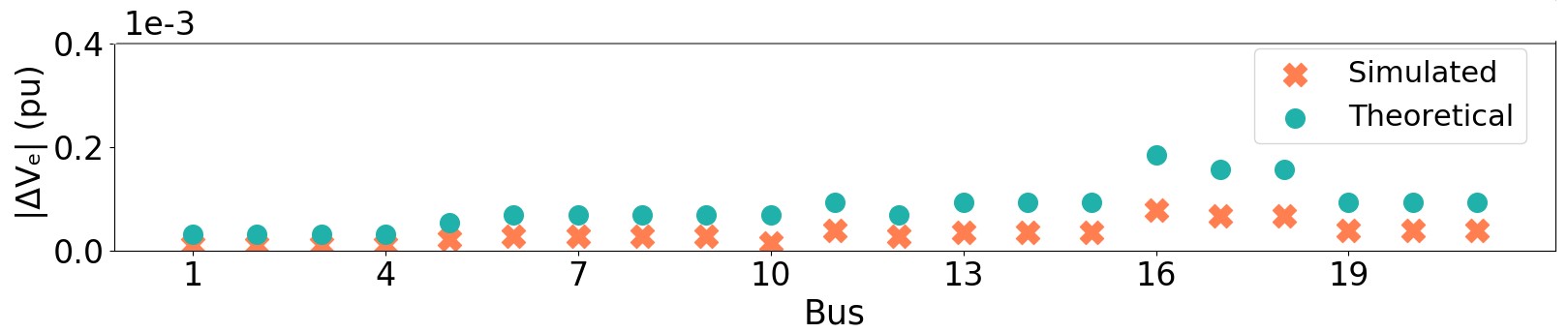}
	\caption{\added{Error bound on the voltage change of phase $a$ in the selected nodes of 123-node network}}
	\label{fig:4b}
\end{figure}

\section{Probabilistic analysis of Voltage sensitivity}
Corollary 1 allows us to compute the voltage change at any observation node from known power changes at various actor nodes. However, in practice, the power could vary randomly due to intermittent characteristics of PV generation. This stochastic variation in turn introduces randomness in the voltage across the network. Under such stochastic scenarios, the grid operator might be interested in predicting the probability of experiencing a voltage violation, i.e., $P(|\Delta V_{O}|> 0.05$ p.u.) so that corrective actions can be taken beforehand. Therefore, it becomes relevant and necessary to derive the probability distribution of the magnitude of voltage change at certain nodes of the distribution grid due to random fluctuations in power at actor nodes. This result is provided by Theorem 2.

\subsection{Computing the probability distribution of $|\Delta V_{O}|$} 
\begin{theorem}
	For an unbalanced radial power distribution system, the probability distribution of voltage change at an observation node $(\Delta V_{O})$ due to random changes in power consumption/injection of actor nodes, corresponds to Nakagami distribution
	\begin{equation}	
	|\Delta V_{O}| \sim \text{Nakagami}(m, \omega)
	\label{eq:24} 
	\end{equation}
	where, shape parameter $m=(\sigma_{r}^2 + \sigma_{i}^2)/\theta$ and
	scale parameter $\omega = \sqrt{\sigma_{r}^2 + \sigma_{i}^2}$. Here, 
	$\theta = 2(\sigma_{r}^4 + \sigma_{i}^4+ 2c^2)/(\sigma_{r}^2 + \sigma_{i}^2) $, $\sigma_{r}^2= C_{R}^T \textstyle\sum_{\Delta S} C_{R}$, $\sigma_{i}^2= C_{I}^T \textstyle\sum_{\Delta S} C_{I}$ and $c$ is the covariance between the real and imaginary part of voltage change. $C_{R}$ and $C_{I}$ are dependent on the shared path impedances and base voltages of the actor nodes, and $\textstyle\sum_{\Delta S}$ is the covariance matrix of complex power change across different actor nodes.
\end{theorem}
\begin{proof}
The change in complex voltage at any observation node due to change in complex power injection/consumption of an actor node can be expressed in terms of real and imaginary components as,

\hspace{2.0cm}$\Delta V_{OA}=\Delta V_{OA}^{r} + j\Delta V_{OA}^{i}$, \\
where, the real part ($\Delta V_{OA}^{a,r}$) and imaginary part ($\Delta V_{OA}^{a,i}$) of voltage change at any phase (here, we use phase $a$ as an example that can be applied to other phases also) of observation node $O$ can be written as	
\begin{equation*}
\begin{split}
\Delta V_{OA}^{a,r} = \sum_{h,u}^{}\frac{-1}{|V_{A}^{h}|}[\Delta P_{A}^{h}(R_{OA}^{u}cos(\omega_{A}) - X_{OA}^{u}sin (\omega_{A})) \\ 
+ \Delta Q_{A}^{h}(R_{OA}^{u}sin(\omega_{A}) + X_{OA}^{u} cos(\omega_{A})) ] \\
\Delta V_{OA}^{a,i} = \sum_{h,u}^{}\frac{-1}{|V_{A}^{h}|}[
\Delta P_{A}^{h}(R_{OA}^{u}sin(\omega_{A}) + X_{OA}^{u} cos(\omega_{A})) + \\
\Delta Q_{A}^{h}(X_{OA}^{u}sin (\omega_{A})-R_{OA}^{u}cos(\omega_{A}))]
\end{split}
\label{eq:25}
\end{equation*}
where $h$ $\epsilon$ $\tilde{H}$ and $u$ $\epsilon$ $\tilde{U}$. The sets $\tilde{H}$ and  $\tilde{U}$ denote different phases, i.e., $a,b,c$, and different phase sequence, i.e., $aa,ab,ac$, respectively. $\Delta P_{A}^{h}$ and $ \Delta Q_{A}^{h}$ are the active and reactive power changes, respectively. $R_{OA}^{h}, X_{OA}^{h}$ are the resistance and reactance of shared path between the observation node $O$ and actor node $A $ from the source node. $V_{A}^{h} $ denotes the base voltage of actor node $A$. 

Using the superposition result of Corollary $1$, the net voltage change at an observation node due to aggregate effect of multiple spatially distributed actor nodes can be written as the sum of changes in voltage at the observation node due to every single actor node as,  
\begin{equation}
\begin{split}
\Delta V_{O}^{a}= \sum_{A} \Delta V_{O A }^{a, r} + \sum_{A} \Delta V_{O A }^{a, i}\\
\end{split}
\label{eq:26}
\end{equation}
Intermittent characteristics of PV injection introduces randomness in the power variation. Here, node power change ($\Delta S$) is modeled as zero mean random vector with covariance matrix $\textstyle\sum_{\Delta S}$. As shown in (\ref{eq:26a}), the notation ($\Delta S$) is a compact vector representing the power change of phases $a$ ($\Delta s^{a}$), $b$ ($\Delta s^{b}$), and  $c$ ($\Delta s^{c}$). In addition, the vector representing power changes of a phase, phase $a$ for e.g., is composed by active and reactive power changes for the corresponding phase of all the nodes.
\begin{equation}
\begin{split}
\Delta S &= [\Delta s^{a} \hspace{0.1cm} \Delta s^{b} \hspace{0.1cm} \Delta s^{c}]^T \\
\Delta s^{a} & = [
\Delta P_{1}^{a} \hdots \Delta P_{n}^{a} \hspace{0.2cm}\Delta Q_{1}^{a} \hdots \Delta Q_{n}^{a}]^T
\end{split}
\label{eq:26a}
\end{equation} 
The distribution of $|\Delta V_{O}|$ can be computed as following:
\vspace{0.2cm}\\
$1$. \textit{Define covariance matrix  $\textstyle\sum_{\Delta S}$}:\\
The covariance matrix $\textstyle\sum_{\Delta S}$ of the complex power change is used to quantify the correlation of power changes among various nodes due to geographical proximity. For nodes that do not have PVs, the variance can be set to zero or equal to the nominal load fluctuation variance. In practice, the covariance structure can be learned using historical data.
\vspace{0.2cm}\\
$2.$ \textit{Compute constant vectors $C_{R}$ and $C_{I}$}: \\
  In this work, the network topology with meta parameters is assumed to be known. Let us define two vectors $C_{R}$ and  $C_{I}$ which can be computed using the following equation.

\begin{minipage}{.05\textwidth}
\begin{align*}	
C_{R} = \begin{bmatrix}
	c_{r}^{aa} \\ c_{r}^{ab} \\ c_{r}^{ac}
\end{bmatrix} \\[4pt] 
C_{I} = \begin{bmatrix}
c_{r}^{aa} \\ c_{r}^{ab} \\ c_{r}^{ac}
\end{bmatrix} 
\end{align*}
\end{minipage}
\begin{minipage}{.35\textwidth}
\begin{equation}
c_{r}^{u} = \begin{bmatrix}
	\frac{-(R_{O 1}^{u}cos(\omega_{1}) - X_{O 1}^{u}sin (\omega_{1}))}{|V_{1}^{a}|} \\[1pt]
	\vdots \\[2pt]
	\frac{-(R_{O n}^{u}cos(\omega_{n}) - X_{O n}^{u}sin (\omega_{n}))}{|V_{n}^{a}|} \\[4pt]
	\frac{-(R_{O 1}^{u}sin(\omega_{1}) + X_{O 1}^{u}cos(\omega_{1}))}{|V_{1}^{a}|} \\[1pt]
	\vdots \\[2pt]
	\frac{-(R_{O n}^{u}sin(\omega_{n}) + X_{O n}^{u}cos (\omega_{n}))}{|V_{n}^{a}|} 
\end{bmatrix}
\end{equation}	
\end{minipage}
where $c_{r}^{u}$ is the constant matrix for a given set of actor nodes and $u$ denotes the self or mutual impedance of the phase $a$ line, i.e., $aa,ab,ac$. Similar matrix exist for $c_{i}^{u}$ but with different values and is omitted for brevity. The compact vectors ($C_{R}^{T}$) and ($C_{I}^{T}$) consist of three components corresponding to three phases. Each such component is composed of ratios between the shared path impedance and rated voltage of the corresponding phase for all the nodes.
\vspace{0.2cm}\\
$3.$  \textit{Compute distribution of $\Delta V_{O}^{r}$ and  $\Delta V_{O}^{i}$}:\\
Voltage change at an observation node due to multiple actor nodes can be expressed as the weighted sum of elements of vector $\Delta S$ as shown by equations (\ref{eq:28}, \ref{eq:29}). Weights are given from the elements of $C_{R}^{T}$ and $C_{I}^{T}$, which represent the ratio of shared path impedance and base voltage of the various actor nodes.  Invoking the Lindeberg-Feller central limit theorem, it can be shown that the weighted sum of the element of $\Delta S$ converges in distribution to a Gaussian random variable. That is, the distribution of $\Delta V_{O}^{r}$ and  $\Delta V_{O}^{i}$ can be expressed as, 
\begin{equation}
\begin{split}
\Delta V_{O}^{a, r} = \sum_{A} \Delta V_{O A }^{a, r}=C_{R}^{T}\Delta S  \overset{D}{\sim} \mathcal{N} (0, C_{R}^T \textstyle\sum_{\Delta S} C_{R} )
\end{split}
\label{eq:28}
\end{equation}
\begin{equation}
\begin{split}
\Delta V_{O}^{a, i}=   \sum_{A} \Delta V_{O A }^{a, i}= C_{I}^{T}\Delta S \overset{D}{\sim} \mathcal{N} (0, C_{I}^T \textstyle\sum_{\Delta S} C_{I} )
\end{split}
\label{eq:29}
\end{equation}
where variance $\sigma_{r}^{2}$ and $\sigma_{i}^{2}$ of $\Delta V_{O}^{a, r}$ and $\Delta V_{O}^{a, i}$ are $C_{R}^T \textstyle\sum_{\Delta S} C_{R}$ and $C_{I}^T \textstyle\sum_{\Delta S} C_{I}$, respectively.
\vspace{0.2cm}\\
$4.$ \textit{Compute distribution of $|\Delta V_{O}|$}: \\
After obtaining the voltage change in terms of the real part $\Delta V_{O}^{r}$ and imaginary part  $\Delta V_{O}^{i}$, the magnitude of voltage change can be written as
\begin{equation}
\begin{split}
|\Delta V_{O}|^{2} = (\Delta V_{O}^{a, r})^2 +  (\Delta V_{O}^{a, i})^2 
\end{split}
\label{eq:30}
\end{equation}
Square of Gaussian random variables follows a gamma distribution as $(\Delta V_{o}^{a,r})^2 \sim \Gamma (0.5, 2\sigma_{r}^{2} )$, $(\Delta V_{o}^{a,i})^2 \sim \Gamma (0.5, 2\sigma_{i}^{2} )$ \cite{lancaster2005chi}.
The shape parameter is $0.5$ and scale parameter is twice the variance of $\Delta V_{O}^{a, r}$, $\Delta V_{O}^{a, i}$ for $(\Delta V_{o}^{a,r})^2$, $(\Delta V_{o}^{a,i})^2$, respectively.
The real and imaginary part of voltage change is correlated with $c=C_R^T\textstyle\sum_{\Delta S} C_I$ as covariance. Then, the covariance between the square terms, i.e., $(\Delta V_{o}^{a,r})^2$ and $(\Delta V_{o}^{a,i})^2$ is $2c^2$. Since, the sum of the correlated gamma variable is also a gamma \cite{chuang2012approximated}, the sum of $(\Delta V_{o}^{a,r})^2$ and $(\Delta V_{o}^{a,i})^2$ follows a Gamma distribution
\begin{equation}
\begin{split}
|\Delta V_{O}|^{2} = |\Delta V_{O}^{a, r}|^{2} +  |\Delta V_{O}^{a, i}|^{2} \sim \Gamma (k, \theta ) 
\end{split} 
\label{eq:32}  
\end{equation}  
where scale parameter $\theta = 2(\sigma_{r}^4 + \sigma_{i}^4+2c^2) /(\sigma_{r}^2 + \sigma_{i}^2) $ and shape parameter $k=(\sigma_{r}^2 + \sigma_{i}^2) /\theta$. The square root of $|\Delta V_{O}|^{2}$ which is a random gamma variable, follows a Nakagami distribution \cite{nakagami1960m}, and therefore the voltage change magnitude will have the following distribution,
\begin{equation}
|\Delta V_{O}| \sim Nakagami(m, \omega),
\label{eq:33}
\end{equation}
where shape parameter $m=k$, scale parameter $\omega=\sqrt{k\theta}$.
\end{proof}
\vspace{-0.2cm}
Theorem 2 is useful in many ways. Using the equation (\ref{eq:33}), the vulnerability of certain observation nodes in terms of voltage violation can be identified quantitatively and efficiently. \added{ Furthermore, one can also leverage the probabilistic framework to find dominant nodes, that have maximum influence on the voltage sensitivity of critical nodes such as hospitals, schools, etc,. Later, the power at these dominant nodes can be controlled to mitigate voltage violations at the critical nodes \cite{jhala2019data}}. Specifically, the vulnerability of nodes in terms of voltage violations can be evaluated by using the probability of the voltage change exceeding a certain threshold ($|\Delta V|>0.05$ p.u.). Further, the proposed PVSA method can be applied to distribution systems with on-load tap changing transformers and voltage regulators with little modifications.
Voltage change/sensitivity needs to be recomputed whenever the tap settings of the transformer changes because the change in substation voltage changes the voltage of all the nodes in the distribution network.
\added{For regulators,}\deleted{In the proposed framework} we can group all the downstream nodes connected to voltage regulator as an independent network and then perform PVSA for the smaller network. For the upstream distribution system, the voltage regulator can be considered as a single node representing a cumulative load of downstream network.
\vspace{-0.2cm}
\subsection{Validation of PVSA for three phase system}
To evaluate the performance of the proposed theoretical approach, we present two case studies using the same IEEE 37\added{-node test} \deleted{bus} system \added{and IEEE 123 node test network }as shown in Fig. \ref{fig:2} and Fig. \ref{fig:2b}, respectively.
\added{In the first case, power is varied randomly on all odd numbered nodes, following Gaussian distribution with zero mean. The assumption of Gaussian distribution is considered as a common assumption
applied in many prior works \cite{hassanzadeh2010practical,vasilj2015pv, jhala2019data}. The covariance matrix $\textstyle\sum_{\Delta S}$ is constructed based on the correlation of power changes on various actor nodes due to their geographical proximities. Note that the proposed approach is quite general and can be applied
to PV generation scenarios with different probability distributions.}
\deleted{In the first case, power is varied randomly on all odd numbered nodes. Due to their geographical proximity, the power changes of various actor nodes may be correlated.}The underlying covariance structure $\textstyle\sum_{\Delta S}$ can be learned from historical or irradiance related data and it's elements are set \added{realistically based on real PV data, and the base loads on the test network are the same as reported in IEEE PES Distribution system analysis subcommittee report.} \deleted{in this work as follows:} For nodes with PV's, the variance of change in real power and reactive power for any phase are set to $50$ kW and $40$ kVar, respectively. The variance of $\Delta P$ and $\Delta Q$ is set to zero for all non actor nodes. 
The off diagonal elements of covariance matrix captures the covariance between different actor nodes, where correlation coefficient between $\Delta P$'s for different actor nodes within the same phase is set to $0.6$ and for $\Delta Q$'s, it is $0.5$. Here, covariance between cross phase terms is assumed to be zero but the proposed approach is quite general to accommodate other covariance structures as well. The correlation coefficient between $\Delta P$'s and $\Delta Q$'s within the same phase is set as $-0.2$. For illustration purpose, the variance of all actor nodes is set to same value, but the values can vary with the nodes depending upon the size and location of PVs. 

The probability distribution of voltage change at node $9$ using two approaches, i.e. the proposed analytical approximation method and the traditional Newton-Raphson based VSA method are plotted in Fig. \ref{fig:5}. For computing the actual distribution of the magnitude of voltage change, a scenario is generated where power is varied randomly on all actor nodes using the above described covariance structure. Then, a change in voltage is computed using NR based sensitivity analysis method. The complete process, i.e., scenario generation and load flow execution is repeated a million times to plot the histogram shown in Fig. \ref{fig:5}. On the other hand, for computing theoretical distribution, the value of vectors $C_{R}$ and $C_{I}$ are calculated using the network parameters. Then, the variance of real ($\Delta V_{O}^{a, r}$) and imaginary ($\Delta V_{O}^{a, i}$) part of voltage change, i.e., $\sigma_{r}^2$ and $\sigma_{i}^2$ are computed by plugging the above defined covariance matrix in equations (\ref{eq:28}, \ref{eq:29}). Finally, the shape and scale parameter of voltage change magnitude, which is a Nakagami distribution, can be directly computed using equation (\ref{eq:33}).
Fig. \ref{fig:5} shows the sufficiently high accuracy of the proposed method \added{particularly the tail probabilities which is our region of interest.} \deleted{, as the} The Jensen-Shannon distance between the actual and theoretical distribution \added{is} $0.07$ \added{where} \deleted{with} $0$ \added{represents} \deleted{representing} identical distribution and $1$ \added{denotes} \deleted{representing} maximally different cases\cite{endres2003new}. \added{The Jensen-Shannon distance average over all the nodes of the network is $0.06$.}
\begin{figure}[h!]
\centering
	\includegraphics[width = 7.0cm,height=4.0cm]{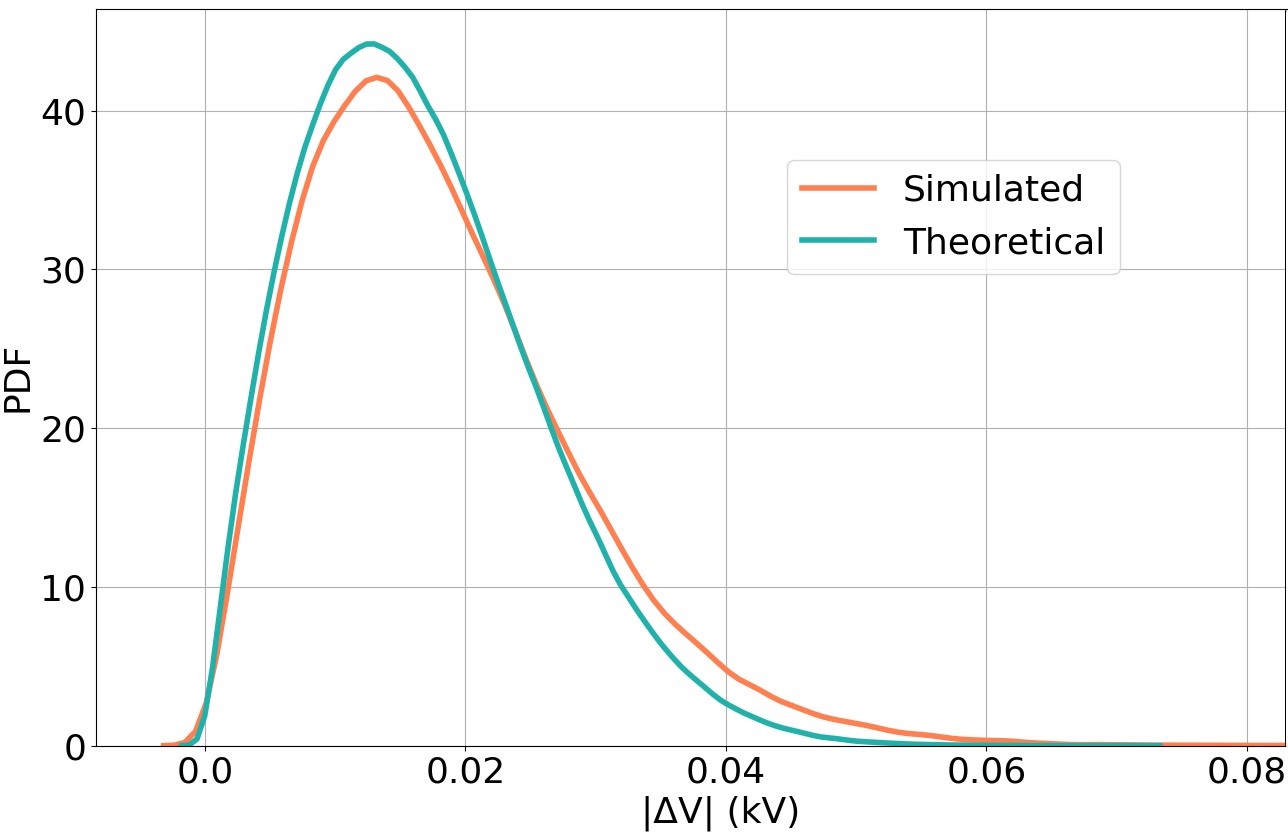}
	\caption{Distribution of magnitude of voltage change at \added{node 9 of the 37 node test network} }
	\label{fig:5}
\end{figure}

\added{Furthermore, the PVSA formulation considering randomness is implemented on the modified IEEE 123-node test system for deriving the distribution of the magnitude of voltage change at node 10, assuming nodes 7, 11, 19, 28, 35, 52, 68 as actor nodes. The covariance matrix is developed in an identical way to a $37$-node network, with the same parameters as discussed in the above paragraph. Fig. \ref{fig:6} depicts the distribution of the magnitude of voltage change computed using the proposed analytical method and the load flow calculation based numerical approach. High accuracy can be witnessed, as the Jensen-Shannon distance between the resulting PDFs of two approaches is $0.18$. This result demonstrates the scalability of the proposed method and its efficacy in conducting VSA for a larger heterogeneous network that includes both single and three phase loads.} 
\begin{figure}[h!]
\centering
	\includegraphics[width = 7.1cm, height=4.0cm]{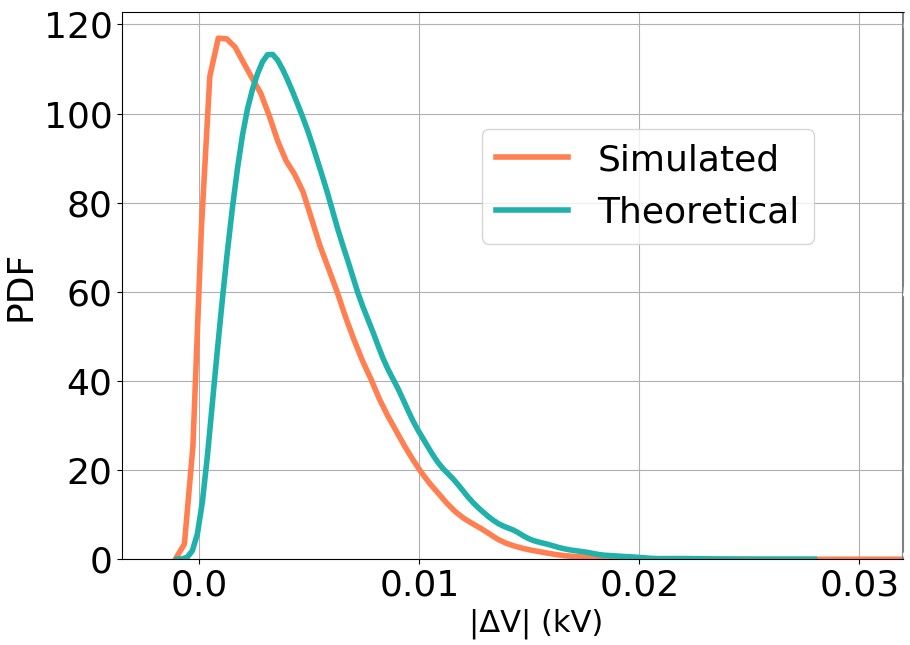}
	\caption{\added{Distribution of magnitude of voltage change at node 10 of the 123-node test network }}
	\label{fig:6}
\end{figure}

\added{The merit of the proposed PVSA method in terms of computational time reduction is demonstrated by comparing the execution time of the simulated and theoretical approaches. All the experiments are conducted on a system with Intel i7 processor running at 2.2 GHz.} The complexity of the proposed
analytical method is of the order $O(1)$, because the calculation
of voltage change in Theorem $2$ does not scale with the size of
the network $(n)$. While, the complexity in NR method is of order $O(n^3)$, as it involves the inversion
operation of the Jacobian matrix. 
\added{Further to compare the execution time for computing voltage change distribution, Monte-Carlo simulations are incorporated to capture the uncertainties associated with the power changes. Here, Monte-Carlo simulations are run for one million times.} Specifically, the execution
time of our method to calculate the voltage change distribution \added{in the 37-node network due to random power changes is within} $1$ minute, compared to $2.2$ hours in classical load flow method.\deleted{with intel i7 processor} \added{On the other hand, the execution time for computing voltage change distribution in the IEEE 123-node system is also within 1 minute, whereas 2.5 hours are needed in the conventional load flow approach. Table \ref{Table:Computationtime} summarizes the time consumption for different cases using the two approaches, respectively. Significant computational time saving can be witnessed with the proposed analytical approach. In particular, the gap is larger in the case of voltage change distribution that includes generator uncertainty. This further highlights the merits of the proposed analytical method, especially in distribution networks with high uncertainty.} \deleted{ Thus, the proposed method is nearly hundred times faster and this factor will increase significantly as the network size increases.}\added{Moreover, the computational time saving offered by the proposed method increases with the size of the distribution network. }                      
\begin{table}[h!]
    \centering
    \caption{\added{Computation time for various case studies}}
    \label{Table:Computationtime}
	\begin{tabular}{|p{28mm}|p{18mm}|p{18mm}|c|}
	\hline
    \backslashbox{Description}{Method} & \textbf{Proposed \newline approach (s)} & \textbf{Load flow \newline approach (s)}  \\
    \hline
    Single observation node on 37 node network & 0.05 & 0.63 \\
    \hline
    Single observation node on 123 node network & 0.09 & 1.57 \\
    \hline
    Distribution for an observation node on 37 node network & 9 & 7920 (1 Million MCS) \\
    \hline
    Distribution for an observation node on 123 node network & 12 & 9200 (1 Million MCS) \\
    \hline
  \end{tabular}
\end{table}


\vspace{-0.1cm}
\section{Conclusion}
This work proposes an analytical approximation of voltage change at any node of the distribution network due to changes in complex power at different actor nodes across a three phase unbalanced distribution network. The approximation error is shown to be tightly upper bounded, illustrating the fidelity of our approach. We also derive the distribution of magnitude of voltage change due to random change in power at actor nodes and show that it can be approximated by a Nakagami distribution. All theoretical results presented in this work are validated with the classical Newton-Raphson load flow method in a modified IEEE 37-node test system \added{and the IEEE 123-node network}. The proposed method can be useful for grid operation and planning as it efficiently allows us to compute the probability of voltage violation at any node in the network. As part of our future work, we will \added{study the impact of different load types on the analytical approximation and} use the probabilistic approach to identify the dominant nodes in the network which have maximum influence on the voltage fluctuation of critical nodes. This information can be useful to develop proactive control strategies for voltage regulation. 

\section*{Acknowledgment}
This material is based upon work partly supported by the Department of Energy, Office of Energy Efficiency and Renewable Energy (EERE), Solar Energy Technologies Office, under Award \# DE-EE0008767 and National science foundation under award \# 1855216.




\bibliographystyle{IEEEtran}
\bibliography{IEEEabrv, PVSA_threeph}
\end{document}